\documentclass[
 reprint,
 amsmath,amssymb,
 aps,
letterpaper]{quantumarticle}
\pdfoutput=1

\usepackage{graphicx}
\usepackage{hyperref}
\usepackage{braket}
\usepackage[linesnumbered, ruled, vlined, titlenotnumbered, noend]{algorithm2e}
\usepackage{amsmath, amsfonts, amssymb, amsthm, color}
\usepackage{multirow} 
\usepackage[capitalise]{cleveref}
\crefname{algocf}{alg.}{algs.}
\Crefname{algocf}{Algorithm}{Algorithms}

\newtheorem{definition}{Definition}
\newtheorem{proposition}{Proposition}

\theoremstyle{remark}
\newtheorem{remark}{Remark}

\newcommand{\q}{\textbf{q}}
\newcommand{\anc}{\textbf{a}}
\newcommand{\ve}{\textbf{v}}
\newcommand{\uu}{\textbf{u}}

\newcommand{\Z}{{\mathbb{Z}}}

\newcommand{\polyA}{{\cal A}}
\newcommand{\polyB}{{\cal B}}
\newcommand{\polyC}{{\cal C}}
\newcommand{\polyP}{{\cal P}}

\newcommand{\rank}{\text{rank}}
\newcommand{\row}{\text{row}}
\newcommand{\matspan}{\text{span}}

\usepackage{xspace}
\newcommand{\CxC}{\textsf{CxC}\xspace}
\newcommand{\CxR}{\textsf{CxR}\xspace}
\newcommand{\CTwo}{\textsf{C2}\xspace}

\begin{document}

\title{Cyclic Hypergraph Product Codes}

\author{Arda Aydin}
\affiliation{
    IonQ Inc.
}
\affiliation{Department of ECE and Institute for Systems Research, University of Maryland, College Park, MD 20742}

\author{Nicolas Delfosse}
\email{nicolas.delfosse@ionq.co}
\affiliation{
    IonQ Inc.
}

\author{Edwin Tham}
\email{tham@ionq.co}
\affiliation{
    IonQ Inc.
}


\begin{abstract}
Hypergraph product (HGP) codes are one of the most popular family of quantum low-density parity-check (LDPC) codes. Circuit-level simulations show that they can achieve the same logical error rate as surface codes with a reduced qubit overhead.
They have been extensively optimized by importing classical techniques such as the progressive edge growth, or through random search, simulated annealing or reinforcement learning techniques.
In this work, instead of machine learning (ML) algorithms that improve the code performance through local transformations, we impose additional global symmetries, that are hard to discover through ML, and we perform an exhaustive search.
Precisely, we focus on the hypergraph product of two cyclic codes, which we call \CxC codes and we study \CTwo codes which are the product a cyclic code with itself and \CxR codes which are the product of a cyclic codes with a repetition code.
We discover \CTwo codes and \CxR codes that significantly outperform previously optimized HGP codes, achieving better parameters and a logical error rate per logical qubit that is up to three orders of magnitude better.
Moreover, some \CTwo codes achieve simultaneously a lower logical error rate and a smaller qubit overhead than state-of-the-art LDPC codes such as the bivariate bicycle codes, at the price of a larger block length.
Finally, leveraging the cyclic symmetry imposed on the codes, we design an efficient planar layout for the QCCD architecture, allowing for a trapped ion implementation of the syndrome extraction circuit in constant depth.
\end{abstract}

\maketitle

\section{Introduction}

Constructing a large-scale quantum computer requires the use of quantum error-correction (QEC) to keep errors in check.
QEC confers robustness by redundantly encoding information across many qubits, in a code whose check operators can be measured without collapsing encoded logical qubits.

Low-density parity-check (LDPC) codes are widely used in classical information processing~\cite{gallager2003low, mackay2003information, richardson2008modern}.
Their generalization to the quantum setting was proposed as early as 2003 in~\cite{mackay2004sparse} which introduced several constructions of quantum LDPC codes.
Unlike surface codes~\cite{kitaev1997quantum, dennis2002topological, raussendorf2007fault, fowler2012surface}, quantum LDPC codes can achieve a constant encoding rate and a growing minimum distance.
However, all the constructions of ~\cite{mackay2004sparse} are limited to a minimum distance that is at most logarithmic in the block length.
In 2009, Tillich and Zémor proposed the HGP construction producing the first family of quantum LDPC codes with constant encoding rate and polynomial minimum distance~\cite{Tillich2014}.
These asymptotic results suggest a path to more efficient QEC but the design of efficient syndrome extraction circuits, layouts and decoders are needed to provide a practical advantage over surface codes.

It took until 2022 to demonstrate that HGP codes outperform surface codes with circuit-level simulations in~\cite{tremblay2022constant} by introducing more efficient syndrome extraction circuits, combined with code optimization techniques and decoders previously proposed in~\cite{grospellier2021combining}.
Moreover, \cite{tremblay2022constant} also proposes a layout based on four planar layers, compatible for example with superconducting qubits equipped with long-range couplers~\cite{marxer2023long, heya2025randomized}.

Several approaches have been explored to optimize HGP codes.
Given that they are built from a pair of classical linear codes, it is natural to select HGP codes through optimization of the classical input codes as in~\cite{connolly2024fast}. This is the role of the progressive edge growth algorithm~\cite{hu2001progressive} which produces high performing classical codes by eliminating local patterns, such as the so-called trapping set that degrade decoder performance. In the quantum setting, the challenge with this approach is that quantum decoders do not generally work in the exact same way as classical decoders and quantum trapping sets are not well understood because decoder design is still rapidly evolving~\cite{raveendran2021trapping, morris2023absorbing}.

The HGP codes of~\cite{grospellier2021combining, tremblay2022constant} were obtained by generating hundreds of random HGP codes and selecting the one achieving the lowest logical error rate for a specific noise rate with code capacity simulations.
The main limitation of that approach is that the code capacity model is too simplistic and it does not guarantee good performance in practice.
Circuit level simulations provide a more accurate indication of the real-world applicability of a code but they consume too much resources to simulate a large number of codes.
Machine learning algorithms were considered in~\cite{Freire2025} where HGP codes are optimized through local modifications of the code structure guided by random walks, simulated annealing or reinforcement learning, which leads to further improvement of previously selected HGP codes. We refer to these codes as ML-optimized HGP codes.

The main challenge in the optimization of HGP codes is the massive size of the search space. In this work, we circumvent this difficulty by imposing additional symmetries to the codes. Namely, we consider {\em cyclic HGP codes} that are the product of two cyclic codes and that we call \CxC codes.
We further restrict our focus by considering {\em symmetric cyclic HGP codes}, denoted $\CTwo$, which are the product of a cyclic code with itself, and {\em repeated cyclic HGP codes}, denoted \CxR, which are the product of a repetition code with a cyclic code.These symmetries sufficiently reduce the search space for us to perform an exhaustive search over all products of cyclic LDPC codes with check weights up to 5 and block length up to 40, resulting in HGP codes with stabilizer weight up to 10 and block length up to 3200.

Our circuit level simulations with physical error rate equal to $10^{-3}$ show that cyclic HGP codes perform far better than the ML-optimized HGP codes.
We found a $[[882, 50, 10]]$ \CTwo code that achieves a logical error rate per logical qubit below $2\times 10^{-8}$ where the ML-optimized $[[625, 25, 8]]$ code of \cite{Freire2025} only achieves $\approx 2\times 10^{-5}$.
Another example is a $[[450, 32, 8]]$ \CTwo HGP code which reaches the same minimum distance as the previous ML-optimized HGP codes with shorter block length and more logical qubits. Moreover, it achieves a logical error rate per logical qubit of $4.5\times 10^{-7}$.
Surprisingly, despite their very simple structure, even the \CxR codes achieve a significantly better logical error rate per logical qubit than the ML-optimized HGP codes.

A recent breakthrough is the discovery of bivariate bicycle (BB) codes which achieves better code parameters and logical error rate than surface codes and previous HGP codes, and can be implemented over two planar layers instead of four~\cite{Bravyi2024}.
Remarkably, \CxC codes achieves comparable performance to the BB codes in terms of logical error rate and qubit overhead.
We even found an example of \CTwo code that achieves simultaneously a lower logical error rate per logical qubit and a smaller qubit overhead than BB codes.
The price to pay with HGP codes is a larger block length.
The advantage of HGP codes might reside in their robustness against hook errors~\cite{manes2025distance}.

The cyclic structure does not only improve code performance but it also simplifies qubit layout.
We use this property to design a layout for \CxC codes over a $2 \times n$ array of qubits equipped with a cyclic shift inspired by the BB code layout of \cite{Tham2025}.
The first row contains the data qubits and the second row contains the ancilla qubits. The cyclic shift lets us align the ancilla qubits with the data qubits they need to interact with.
We obtain a constant depth syndrome extraction for \CxC codes that are LDPC.

The $2 \times n$ model, introduced in~\cite{siegel2024towards} and extended in~\cite{Tham2025}, is a special case of the QCCD architecture~\cite{kielpinski2002architecture} where qubits are placed on a $2 \times n$ array and qubit moves are limited to a cyclic shift, making our cyclic HGP codes practically relevant for quantum computing platforms with flying qubits capable of implementing a cyclic shift such as photonic qubits~\cite{knill2001scheme}, spin qubits~\cite{loss1998quantum}, electron on liquid helium~\cite{lyon2006spin}, trapped ions~\cite{cirac1995quantum} and neutral atoms~\cite{bluvstein2024logical}.
Our codes are also compatible with other HGP code layouts such as the spin qubit layout of~\cite{siegel2024towards}, the neutral atom layout of~\cite{xu2024constant, pecorari2025high} or the 2D local layouts of~\cite{delfosse2021bounds, berthusen2024partial} or the modular layout of~\cite{strikis2023quantum}.

The rest of this paper is organized as follows. \cref{sec:background} reviews relevant quantum LDPC code constructions. See~\cite{breuckmann2021quantum} for a more thorough review.
Cyclic HGP codes are studied in \cref{sec:cyclic_HGP_codes} and a layout for these codes is proposed in \cref{sec:Layout}.

\section{Background}
\label{sec:background}
{\em Stabilizer codes}~\cite{gottesman1997stabilizer} are QEC codes whose check operators are elements of a {\em stabilizer group}, $\mathcal{S}$.
That group's {\em stabilizer generators}, $S$, in turn are a set of mutually commuting Pauli operators: $\mathcal{S} = \langle S \rangle$, $S=\{s_1,...,s_n\}$, where $s_i\in \{I,X,Y,Z\}^n$ and $s_i s_j = s_j s_i$ for all $1\leq i,j \leq n$.

{\em CSS codes}~\cite{calderbank1996good, steane1996multiple} are stabilizer codes where the set of stabilizer generators can be partitioned into $X$ and $Z$ types, whose only non-trivial Pauli operators are $X$ and $Z$ respectively: $S = S_X \cup S_Z$, where each $s\in S_P$ is of the form $s\in\{I,P\}^n$, for $P\in \{X,Z\}$.
On a CSS code with $m_P = |S_P|$ stabilizer generators of each type, it is convenient to organize them into rows of parity check matrices $H_P \in \mathbb{F}_2^{m_P \times n}$, s.t. $\left[H_P\right]_{i,j}$ is 1 if $s_{P,i}\in S_P$ acts non-trivially on the $j$-th qubit, and 0 otherwise.
These matrices are so called, because the codewords of the code in the $X$ ($Z$) basis belong in $\ker{H_X}$ ($\ker{H_Z}$) respectively.
Requiring $H_X \cdot H_Z ^T = 0$ is equivalent to having stabilizer generators that mutually commute.

CSS codes may be constructed by elevating two compatible classical codes through the two-block construction~\cite{Kovalev2013}.
Given classical codes with parity check matrices (sometimes called ``seed matrices'') $A\in \mathbb{F}_2^{m_a \times n_a}$ and $B\in \mathbb{F}_2^{m_b \times n_b}$, a corresponding quantum CSS code might be constructed by defining $H_X = [A | B^T]$ and $H_Z = [B | A^T]$.
The two classical codes are said to be {\em compatible} when $m_a=n_a=m_b=n_b$ ({\em i.e.} $A$ and $B$ are square), and when $A B^T + B^T A = 0$ to ensure commutativity of the stabilizer generators.

{\em Hypergraph product} or HGP codes, are quantum CSS codes that are also constructed by elevating two classical codes.
Given classical codes $A$ and $B$ as before, and denoting by $I_n$ a $n\times n$ identity matrix, the corresponding HGP code is defined by: $H_X = [A\otimes I_{n_b} | I_{m_a}\otimes B^T]$ and $H_Z = [I_{n_a}\otimes B | A^T \otimes I_{m_b}]$, with the resulting parity check matrices having dimensions $H_X\in \mathbb{F}_2^{m_a n_b\times (n_a n_b+m_a m_b)}$, $H_Z \in \mathbb{F}_2^{n_a m_b \times (n_a n_b+m_a m_b)}$.
Note that under the HGP construction, there is far greater flexibility since $A$ and $B$ are no longer constrained; \emph{any} two classical codes will suffice.

A common building block for classical and quantum binary codes are sparse {\em cyclic matrices} $C$, which are $m\times n$ binary matrices with the form $C_{i,j} = f(j-i\pmod n)$, for some generating function $f:\mathbb{Z}_n \to \{0,1\}$ and $1\leq i,j\leq n$.
{\em Circulant matrices}, $Q_n^\ell$, are special cyclic matrices of dimension $n\times n$, wherein $f(j)=1$ iff $j=\ell$, and is 0 otherwise.

Several related quantum codes have similarly been built with sparse cyclic matrices.
The first such example, so-called bicycle codes by MacKay {\em et al.}~\cite{Mackay2004}, follows a standard two-block construction, with $A=B^T=C$ for some randomly chosen sparse cyclic matrix $C$.
Generalized bicycle (GB) codes, later introduced by Kovalev and Pryadko~\cite{Kovalev2013}, allowed for $A$ and $B$ to be distinct sparse cyclic matrices.
In that same work hyper-bicycle codes were introduced, which are in fact HGP codes, wherein $A$ and $B$ are sums of tensor products of arbitrary binary matrices with permutations of circulant matrices.
More recently, bivariate bicycle codes by Bravyi et al.~\cite{Bravyi2024}, introduced yet another generalization of the GB codes in which $A$ and $B$ are constructed as sums of tensor products of circulant matrices.

\section{Cyclic Hypergraph Product Code}
\label{sec:cyclic_HGP_codes}
Let us now describe our construction of hypergraph product codes built atop cyclic matrices, which we call the {\em cyclic hypergraph product code}. We refer to these codes as the \CxC codes.

\begin{definition} [The \CxC codes]
A cyclic hypergraph product code is an HGP code, parameterized by positive integers $a$ and $b$ along with univariate polynomials, $\polyA(x)$ and $\polyB(y)$, whose coefficients are in $\mathbb{F}_2$.
Therein $x$, $y$ are circulant matrices $Q_a$, $Q_b$ respectively.
The \CxC code has parity check matrices $H_X = [\polyA\otimes I_b | I_a \otimes \polyB]$ and $H_Z = [I_a \otimes \polyB^T | \polyA^T\otimes I_b]$.
\label{def:cHGP}
\end{definition}

In other words, a \CxC code is the hypergraph product of two cyclic codes.
The \CxC code inherits convenient attributes common to HGP codes generally.
Its block length can be immediately deduced to be $n=2ab$.
Further, let's denote by $r_a$, $r_b$ the ranks of $\polyA$ and $\polyB$  respectively; denote by $d_M$ the minimum distance of the classical code with parity check matrix $M$; and by $w(\polyP)$ the number of summands in polynomial $\polyP$.
Then, the \CxC code encodes $k=2(a-r_a)(b-r_b)$ logical qubits, with minimum distance $d_{\CxC}=\min_{M\in\{\polyA, \polyB\}} d_M$.
These follow directly from previous results for general HGP codes~\cite{Tillich2014}, and the fact that the classical codes whose parity-checks matrices are cyclic matrices $\polyA$ and $\polyB$ are equivalent to the transpose codes associated with $\polyA^T$ and $\polyB^T$, up to reversal of bit labels.

\begin{remark}
The hyper-bicycle code of~\cite{Kovalev2013} introduced a construction that seemed related to ours.
However, each classical code in its HGP construction, is itself a tensor product of the form $\sum_i C_i \otimes M_i$, which introduces added structure (via cyclic matrices $C_i$) on a set of arbitrary matrices $M_i$.
We eschew the additional tensor product, and instead impose cyclicality directly on $M_i$.
\end{remark}

\begin{remark}
The BB codes of~\cite{Bravyi2024,Ye2025} depart from a HGP construction by replacing $\polyA(x)\otimes I$ by $\polyP_1(x,y)$ and $I\otimes \polyB(y)$ by $\polyP_2(x,y)$, and then following a standard CSS construction.
Therein, $\polyP_1, \polyP_2$ are bivariate polynomials in $x,y$, with coefficients in $\mathbb{F}_2$.
\end{remark}

\begin{remark}
Our construction includes well-known members of HGP codes.
Notably, the standard $[[2d^2,2,d]]$ toric code is obtainable by setting $\polyA(x) = 1+x$ and $\polyB(y) = 1+y$, and $a=b=d$.
\end{remark}

\begin{remark}
The La-cross codes of \cite{Pecorari2025} follow a similar construction.
Specifically, it is related to \CxC codes with $a=b$ and generating polynomials $\polyA(x)=\polyB(x)=1+x+x^k$ for a range of $k$'s.
But the authors opted to reduce block length by using only the minimal set of independent rows of $\polyA$ and $\polyB$.
Then, the resulting parity check matrices have full row rank, the codes corresponding to their transpose is trivial, and the number of logical qubits encoded in the La-cross codes is reduced by a factor two.
\end{remark}

\begin{proposition}\label{prop:BalanceCriteria}
The \CxC code has equal numbers of $X$ and $Z$ stabilizer generators, $\rank(H_X) = \rank(H_Z) =  a r_b + b r_a - r_a r_b$.
Moreover, each stabilizer generator is of equal weight $\omega$, given by $\omega=w(\polyA) + w(\polyB)$.
\end{proposition}

\begin{proof}
Recall, $\polyA$ and $\polyB$ take the form $Q^{\ell_1}+Q^{\ell_2}+...$.
Circulant matrices $Q^{\ell_1}$ and $Q^{\ell_2}$ are completely non-overlapping if $\ell_1\neq\ell_2$, and are identical otherwise.
Since each $Q^{\ell}$ is a permutation matrix, each row of $\polyA\otimes I$, $\polyA^T\otimes I$ (resp. $I\otimes \polyB$, $I\otimes \polyB^T$) has a number of non-zero entries equal to the number of distinct summands in $\polyA$ (resp. $\polyB$).
Further, cyclicality of $\polyA$ implies $\matspan(\row(\polyA))$ and $\matspan(\row(\polyA^T))$ are identical, up to reversal of bit labels (likewise for $\polyB$ and $\polyB^T$), which establishes $\rank(H_X) = \rank(H_Z)$.
Finally, since the block length ($n=2ab$) and logical dimension ($k=2(a-r_a)(b-r_b)$) of the \CxC are known, by symmetry we conclude that $\rank(H_X) = \rank(H_Z) = (n-k)/2$.
\end{proof}

We view the balance criteria of \Cref{prop:BalanceCriteria} as an important attribute of the \CxC code.
The balance criteria imbues the \CxC code with near-identical performance in both $X$ and $Z$ logical bases, something that is not generally guaranteed for an arbitrary HGP code.
In the following sub-sections, we showcase explicit \CxC instances we built, along with numerical estimates of their performance.
Then, in the next section, we argue that our \CxC codes admit convenient implementation under a {\em cyclic layout}, constructed via ideas we previously described in~\cite{Tham2025}.

\subsection{Explicit Constructions}\label{subsec:ExplicitConstructions}
In this work, we searched for \CxC codes by considering various choices of $a$, $b$, $\polyA$, and $\polyB$.
Conveniently, our search was aided by the form of seed matrices of the \CxC, which reduces the space of candidates considerably compared to searching among random binary matrices.

Through brute-force computer search, we exhaustively enumerated all classical cyclic codes ($\polyC$) of length $n_c \leq 40$, and with generating polynomials containing $w(\polyC)$ summands, with $2\leq w\leq 5$.
The limits on $w(\polyC)$ were chosen as such because $w(\polyC)=1$ cyclic matrices are always full rank and because we are interested in codes with low-weight checks (noting that $w(\polyC)=5$ already leads to HGP codes with check weight $10$ in the worst case).
The limit on classical block length of $n_c\leq40$ is a practical one, both because the exhaustive search quickly becomes expensive, and because that limit already encompasses resulting \CxC codes with quite large block lengths.

For each candidate classical cyclic code with $w(\polyC)=w$, we reject those with minimum distance $d_c<2$ or which contain $k_c<1$ logical qubits.
Of the remaining classical codes we retained those with higher code rate ($k_c/n_c$).
The results are lists $\tilde{C}_w=\{\polyC(d_c)~|~d_c=2,3,...\}$ of classical cyclic codes of various minimum distances, each with preferential code rates for the given $d_c$.
From these lists, we construct the following two sub-families of \CxC codes.

\begin{definition}[The \CTwo codes]
A symmetric cyclic HGP code is a cyclic HGP code where $a=b$, $x=y$, and $\polyA(x) = \polyB(y)$.
\end{definition}

The \CTwo codes (think C-square) are Cartesian products of a cyclic code $\polyC(d_c)\in\tilde{C}_w$ with itself.
Given an $[n_c,k_c,d_c]$ code from $\tilde{C}_w$, the resulting \CTwo codes have block length $n_{\CTwo}=2n_c^2$, code rate $k_{\CTwo}/n_{\CTwo} = (k_c/n_c)^2$, minimum distance $d_c$, and check weights $\omega=2w$.

\begin{definition} [The \CxR codes]
A repeated cyclic HGP code is a cyclic HGP code where $\polyB(y) = 1+y$.
\end{definition}

The \CxR codes are Cartesian products of a code $\polyC(d_c)\in\tilde{C}_w$ with a repetition code, the simplest example of a cyclic code.
An example of \CxR code was previously considered in~\cite{siegel2024towards}. Here, we thoroughly investigate the performance of this code family.
A nuance specific to \CxR codes is that we choose $b = d_c$; {\em i.e.} since the repetition code $\polyB(y)=1+y$ has parameters of the form $[[n_b=d_b, 1, d_b]]$, it must be chosen with $d_b=d_c$, so as not to drag down the resulting HGP code's minimum distance from what is afforded by the optimized code in $\tilde{C}_w$.
The resulting \CxR code then has block length $n_{\CxR}=2n_c d_c$, code rate $k_{\CxR}/n_{\CxR}=k_c/(n_c d_c)$, minimum distance $d_c$, and check weights $\omega=w+2$.

\begin{figure}
    \begin{centering}
    \includegraphics[width=8.5cm]{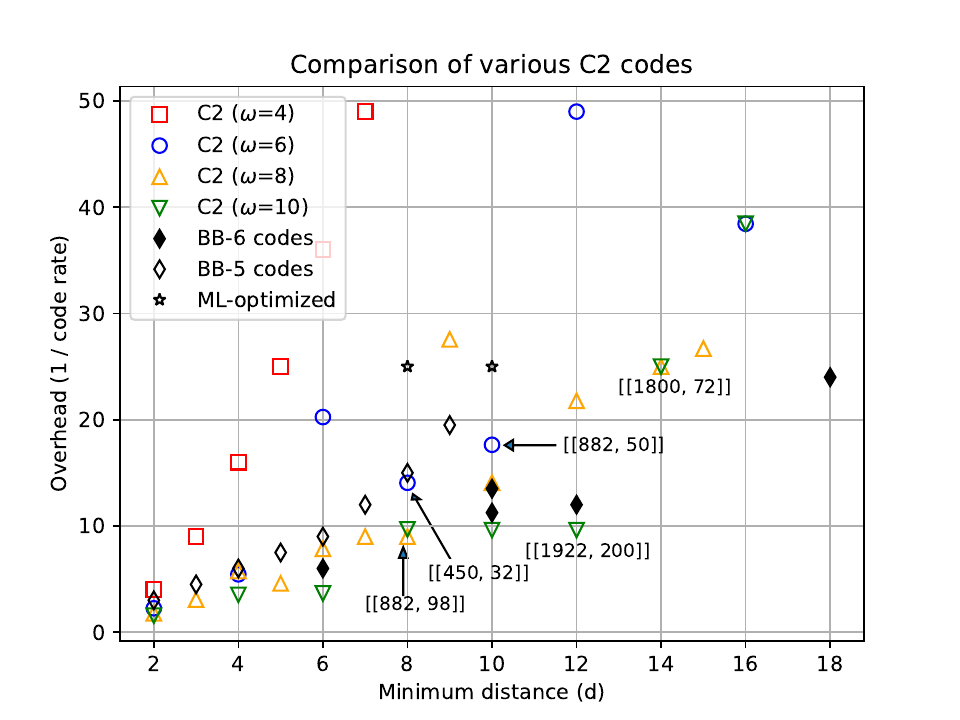}
    \par\end{centering}
    \caption{Comparison of C2 codes of various check weights ($\omega$). Also shown are BB codes (diamonds)~\cite{Bravyi2024,Ye2025} and the ML-optimized HGP codes of~\cite{Freire2025} (star).}
    \label{fig:sym-cHGP}
\end{figure}

\begin{figure}
    \begin{centering}
    \includegraphics[width=8.5cm]{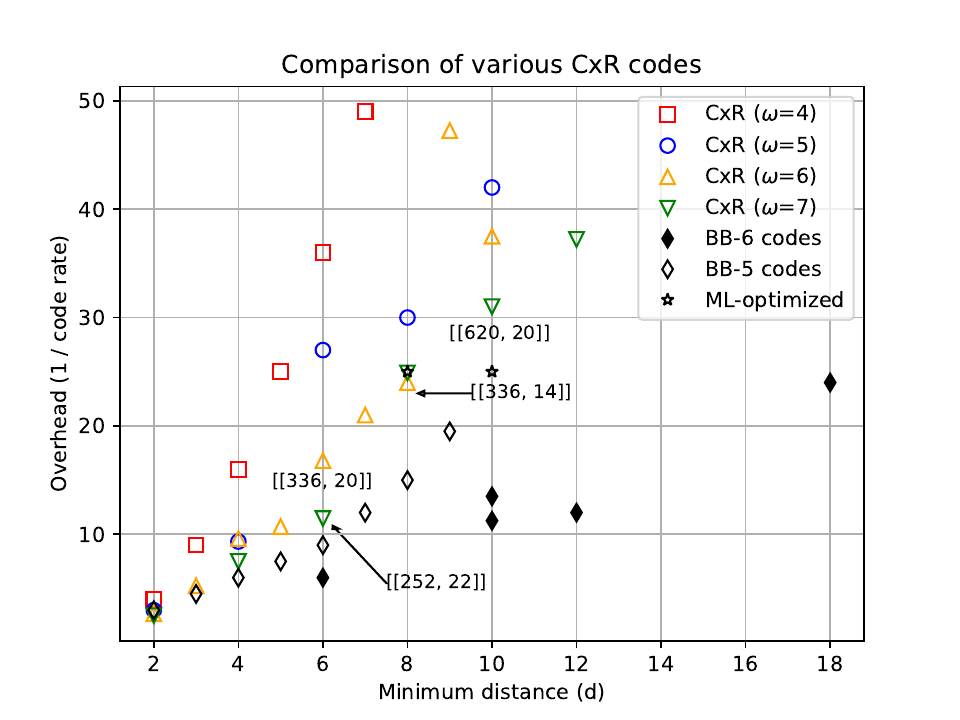}
    \par\end{centering}
    \caption{Comparison of \textsf{CxR} codes of various check weights ($\omega$). Also shown are BB codes (diamonds)~\cite{Bravyi2024,Ye2025} and the ML-optimized HGP codes of~\cite{Freire2025} (star).}
    \label{fig:rep-cHGP}
\end{figure}

\Cref{fig:sym-cHGP} shows a comparison of distance vs encoding overhead (inverse of code rate) for several \CTwo instances, stratified by check weights.
Also shown are BB-5 and BB-6 codes of \cite{Ye2025,Bravyi2024} which have check weights of five and six respectively, as well as the ML-optimized HGP codes of~\cite{Freire2025} optimized using random walks, simulated annealing and reinforcement learning, that have non-uniform check weights averaging $\approx 7$.
A similar comparison for \CxR codes is shown in \cref{fig:rep-cHGP}.

It is noteworthy that \CTwo codes exhibit minimum distances and code rates that are comparable to the BB codes, and significantly beating the ML-optimized HGP codes of \cite{Freire2025}.
This is especially true of \CTwo codes with higher check weights (though we note that higher $\omega$ may lead to worse circuit level performance).
The \CxR codes exhibit poorer parameters, which is to be expected -- most of the classical codes in $\tilde{C}_w$ are such that $n_c/k_c < d_c$, so for a given $\polyC(d_c)\in\tilde{C}_w$ the associated \CTwo instance will usually have much better code rates than its \CxR counterpart.
However, we were interested in \CxR as a sub-family primarily because the simplicity of its construction allows for an especially efficient layout.
Still, our \CxR instances have parameters that are competitive with the ML-optimized codes; and since they generally have checks of lower weight, they offer strong circuit level performance.

\subsection{Performance}\label{subsec:Performance}

To evaluate circuit level performance, we simulated a selection of \CxC codes under standard circuit noise -- wherein two-qubit controlled-Pauli and single-qubit Clifford gates, measurements, resets, and idle qubits all are followed by one- or two-qubit depolarizing channels (whichever is appropriate for the operation) with noise rate $p$.
Operations supported on non-overlapping qubits may occur simultaneously.
The precise circuit we use is described in \cref{sec:Layout}.
Our simulations were implemented in Stim~\cite{Gidney2021stim} and decoded using BP+OSD~\cite{panteleev2021degenerate, RoffeBPOSD2020,RoffeStimBPOSD2022}, instantiated for 10,000 BP iterations and 5-th order OSD.
The logical error rates $\tilde{p}_{\log}$ we report here are normalized by number of syndrome rounds.

In \cref{fig:LogicalComparison} we show a comparison of several \CTwo and \CxR codes against BB-6 and an ML-optimized HGP code.
In order to compare codes with different encoding rates, we plot on the horizontal axis the logical error rate normalized by both number of syndrome rounds and number of encoded logical qubits ({\em i.e.} $\tilde{p}_{\log} / k$).
We used noise rates $p=10^{-3}$ and $3\times10^{-3}$ (shown in \cref{fig:LogicalComparison} as solid vs hollow symbols respectively).

\begin{figure}
    \begin{centering}
    \includegraphics[width=8.5cm]{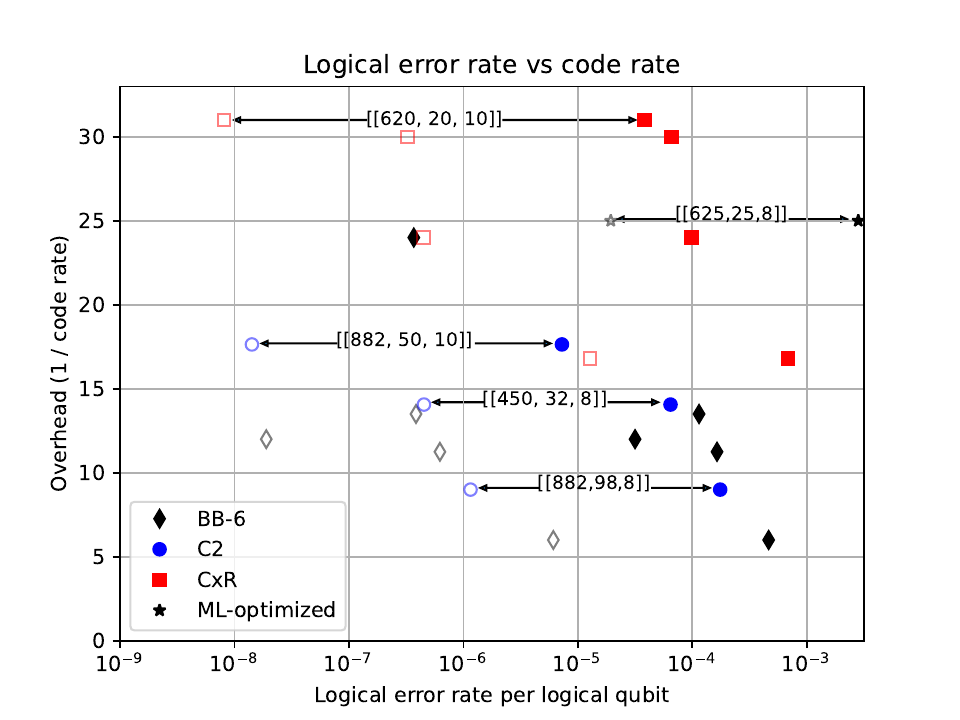}
    \par\end{centering}
    \caption{Circuit level performance versus encoding overhead, under standard circuit noise with rates $p=3\times 10^{-3}$ (solid) and $p=10^{-3}$ (hollow).
    A subset of \CTwo (circles) and \CxR (square) codes are compared against BB-6 (diamond) and ML-optimized HGP (star) codes.}
    \label{fig:LogicalComparison}
\end{figure}

\Cref{fig:sym_circuitperf,fig:rep_circuitperf} similarly show circuit level performance for the same selection of \CTwo and \CxR codes, again under standard circuit noise, but across a broader range of physical noise rate $p$.
Therein, lines represent the logical error rate heuristic $\tilde{p}_{\log} = p^{d/2} e^{\alpha +\beta p + \gamma p^2}$ for real valued parameters $\alpha,\beta,\gamma$ fitted to the numerical data points.
The fitted values for those parameters are listed in Appendix~\ref{appendix:FitParams}.
Also shown for comparison in dashed lines are heuristics for surface codes, $p_{\log}=0.1 (100p)^{(d+1)/2}$, of comparable distance.

In the previous section, we observe that the \CTwo codes achieve better parameters than the ML-optimized HGP codes. This advantage persists when considering the logical error rate. \cref{fig:LogicalComparison} shows that \CTwo codes achieves a logical error rate several orders of magnitude lower than the logical error rate of the ML-optimized HGP codes.
The \CxR codes, while they have worse encoding rates than the ML-optimized HGP codes, nevertheless still boast better logical error rate.

Moreover, the \CTwo codes we found are competitive compared to BB-6 codes in some regimes.
For some target logical error rate, we observe that \CTwo codes achieve that target while having a smaller overhead than BB-6 codes.
Further, in some regimes, \CTwo codes can simultaneously reach a lower logical error rate per qubit and a smaller overhead than BB-6 codes.

\begin{figure}
    \begin{centering}
    \includegraphics[width=8.5cm]{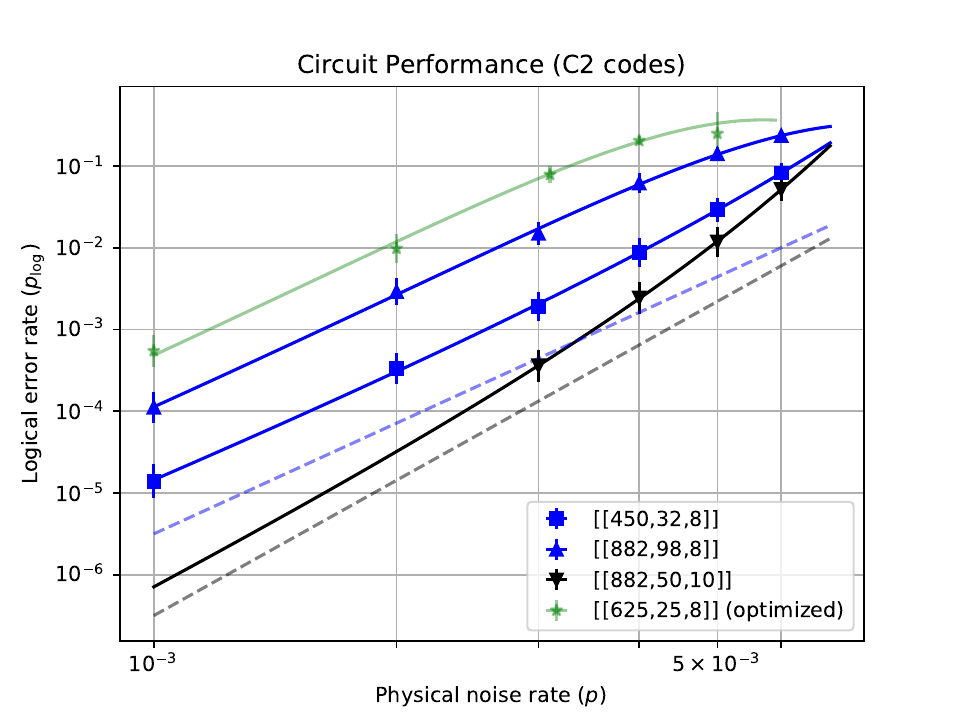}
    \par\end{centering}
    \caption{Circuit level performance for a selection of \CTwo codes, under standard circuit noise.
    Also shown is the $[[625,25,8]]$ ML-optimized HGP code of \cite{Freire2025} (green, star). Dashed lines are surface code heuristics, for comparison.}
    \label{fig:sym_circuitperf}
\end{figure}

\begin{figure}
    \begin{centering}
    \includegraphics[width=8.5cm]{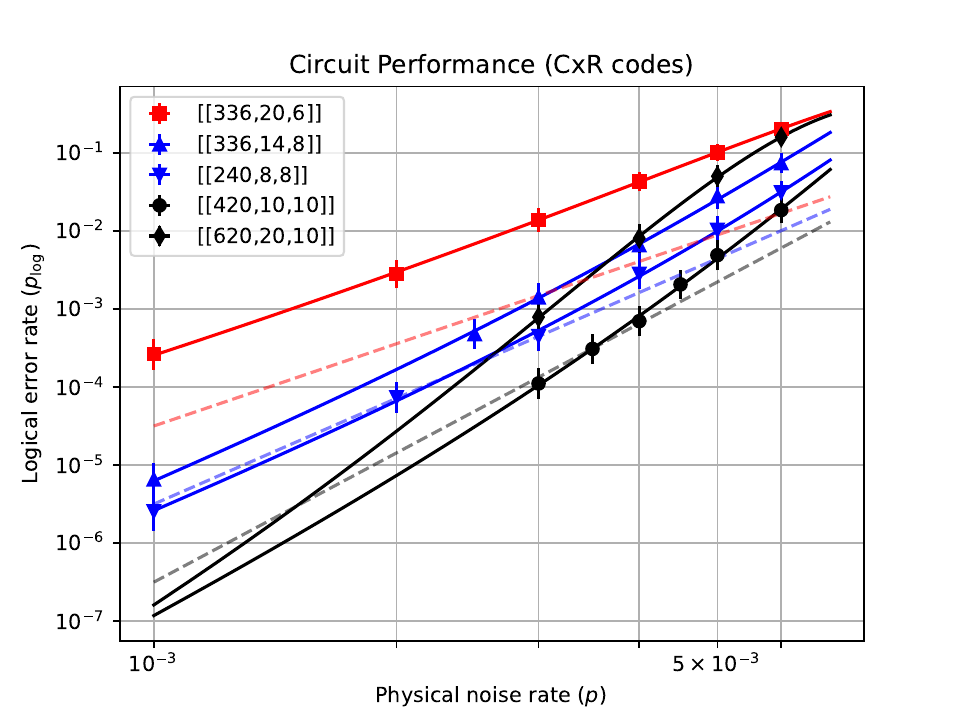}
    \par\end{centering}
    \caption{Circuit level performance for a selection of \CxR codes, under standard circuit noise. Dashed lines are surface code heuristics, for comparison.}
    \label{fig:rep_circuitperf}
\end{figure}

\section{Layout}
\label{sec:Layout}
Our consideration of the sub-family of \CxR codes is strongly motivated by the possibility of an ultra efficient modular layout.
A similar layout for \CxR codes, where the data qubits are moving instead of the ancilla qubits, is considered in~\cite{siegel2024towards}.
We follow the layout strategy introduced in \cite{Tham2025}, which we will not repeat here, except to state that it considers flying qubits that can move relative to each other as a {\em modus operandi} for scaling, and operations between otherwise distant qubits (or sets thereof) are only possible if they are moved and brought into alignment with each other.
Such a model is a realistic abstraction of many quantum hardware platforms, including trapped ions that can be re-organized through ion shuttling.
Effectively, the model segments the many gates of a quantum LDPC code's syndrome circuit into a hierarchy of slow operations like long-range qubit transport (which we work to minimize) vs fast ones like gate or short-range transport operations (of which there can be many).
As we shall see, the layout we design has a cost that directly depends on the number of summands in polynomials $\polyA$ and $\polyB$, and thus favors simple codes like \CxR codes.

To describe such a scalable layout let's denote by $\q_{\uu}$ and $\anc_{\ve}$ the qubits and ancillae of a \CxC code.
Let the qubits be indexed by tuples such that $\uu=(i,j,k)\in \Z_2\times \Z_a\times \Z_b$.
Then, the $u$-th column of the \CxC parity check matrices is related to the qubit $\q_{\uu}$ via the map $u = abi + bj + k$.
Similarly, let ancillae $\anc_{\ve}$ be indexed by tuples such that $\ve=(r,s,t)\in \{X,Z\} \times  \Z_a\times \Z_b$.
Then, the ancilla $\anc_{\ve}$ measures the $v$-th row of parity check matrix $H_r$, with $v=bs+t$.
For convenience, we shall directly refer to rows and columns of \CxC parity check matrices as being indexed by the 2-tuple $(s,t)$ and 3-tuple $(i,j,k)$ respectively.

For the layout of our \CxC, we contemplate an arrangement of qubits and ancillae in two rows -- for the present discussion, let the {\em upper row} contain qubits and {\em lower row} contain ancillae.
Furthermore, let qubits in the upper row be arranged in ascending order of $u$ and let ancillae in the lower row be arranged with $X$ ancillae first followed by $Z$ ancillae with each set in turn being in ascending order of $v$ (with $u$ and $v$ defined above).

Now let the lower row be imbued with a {\em cyclic shift} operation, parameterized by the 3-tuple $(\chi,\eta,\zeta)$.
We define a $(\chi,\eta,\zeta)$ cyclic shift as the permutation of indices in $\ve$ (modulo $2$, $a$, and $b$ respectively) that realizes the following alignment: $\anc_{(X,*,*)}\leftrightarrow\q_{(\chi,*\oplus\eta,*\oplus\zeta)}$, and $\anc_{(Z,*\oplus\eta,*\oplus\zeta)}\leftrightarrow\q_{(1\oplus\chi, *,*)}$.
(Here $\oplus$ denotes modular addition, with the obvious modulus by context, and $*$ represents the range of all possible values for the given index.)
Physically, such a cyclic shift is achievable by cyclicly translating ancillae in the lower row, over distances $ab\chi$, $b\eta$, and $\zeta$, with periods $2ab$, $ab$, and $b$ respectively; see~\cite{Tham2025} for a discussion on the physical implementation of the cyclic shift using ions.

We note, that the qubit/ancilla alignment implied by a $(\chi,\eta,\zeta)$ cyclic shift requires physical cyclic translation of the $r=X$ vs $r=Z$ ancillae in opposite chiral senses -- {\em i.e.} a ``clockwise'' cyclic translation of $(X,*,*)$ ancillae must be accompanied by a corresponding ``anti-clockwise'' one on $(Z,*,*)$ ancillae.
Further, in practice an implementation of the cyclic shift may choose to decompose it into multiple steps -- for instance by first applying $(\chi,0,0)$ followed afterward by the $(\chi,\eta,\zeta)$ cyclic shifts -- out of Physics or engineering considerations.
In this example, the intermediate $(\chi,0,0)$ cyclic shift implies longer transport distances given our qubit arrangement, and may therefore necessitate special cooling steps.

\begin{algorithm}
\DontPrintSemicolon
\SetAlgoLined

\KwIn{A \CxC code and an integer $d$.\;}
\KwOut{A circuit measuring $d$ rounds of all stabilizer generators of the input code.\;}
Prepare all ancillae in the state $\ket +$.\;

\For{$0\leq l \leq d$}{
    \For{each monomial $x^{\eta} \in \polyA(x)$}{
        Do $(0,\eta,0)$ cyclic shift.\;
        Apply CX gates between ancillae $(X,*,*)$ and qubits $(0,*\oplus \eta,*)$, if $l<d$.\;
        Apply CZ gates between ancillae $(Z,*\oplus \eta,*)$ and qubits $(1,*,*)$, if $l>0$.\;
    }
    Measure and reset ancillae $(Z,*,*)$ in $X$ basis, if $l>0$.\;

    \If{$l<d$}{
    \For{each monomial $y^{\zeta} \in \polyB(y)$}{
        Do $(1,0,\zeta)$ cyclic shift.\;
        Apply CX gates between ancillae $(X,*,*)$ and qubits $(1,*,*\oplus \zeta)$.\;
        Apply CZ gates between ancillae $(Z,*,*\oplus \zeta)$ and qubits $(0,*,*)$.\;
    }
    Measure and reset ancillae $(X,*,*)$ in $X$ basis.\;
    }
}

\caption{\CxC codes sparse cyclic layout}
\label{algorithm:sparse_cyclic_layout}
\end{algorithm}

\begin{proposition}\label{prop:SparseCyclicLayout}
    \Cref{algorithm:sparse_cyclic_layout} implements a syndrome extraction circuit for a \CxC code, with two rows of qubit modules imbued with cyclic shifts.
    The resulting circuit has depth $(2w(\polyA)+2w(\polyB)+2)d + (2w(\polyA)+1)$.
\end{proposition}

In this proposition, $d$ is the number of rounds of syndrome extraction executed. Therefore, the amortized depth per round of syndrome extraction converges to $(2w(\polyA)+2w(\polyB)+2)$ including
$w(\polyA)+w(\polyB)$ rounds of two-qubit gates, $w(\polyA)+w(\polyB)$ rounds of cyclic shifts, and $2$ rounds of preparation and measurement of ancilla qubits.

\begin{proof}
    Following Lemma~1 in~\cite{Tham2025}, the matrix corresponding to $x^{\ell} y^m$ (a monomial summand of $\polyA$ or $\polyB$) has a non-zero entry on row $(s,t)$ and column $(i,j,k)$ if and only if $(j,k)=(s\oplus\ell,t\oplus m)$, and $i$ takes the value appropriate for whether $x^{\ell}y^m$ appears in $\polyA$ or $\polyB^T$ ($i=0$), or $\polyB$ or $\polyA^T$ ($i=1$).
    
    For a \CxC code, $\polyA(x)$ (resp. $\polyA(x)^T$) contains only terms of the form $x^{\ell}$ (resp. $x^{-\ell}$), so they connect $\q_{(0,*\oplus\ell,*)}\leftrightarrow\anc_{(X,*,*)}$ (resp. $\q_{(1,*,*)}\leftrightarrow\anc_{(Z,*\oplus\ell,*)}$).
    These gates are reflected in Lines~5-6 of \cref{algorithm:sparse_cyclic_layout}.
    Similarly, $\polyB(y)$ (resp. $\polyB(y)^T$) for a \CxC code contains terms of the form $y^m$ (resp. $y^{-m}$), and connect $\q_{(1,*,*\oplus m)}\leftrightarrow\anc_{(X,*,*)}$ (resp. $\q_{(0,*,*)}\leftrightarrow\anc_{(Z,*,*\oplus m)}$).
    The respective gates are reflected in Lines~11-12.

    All gates of Lines~5-6 are non-overlapping and can execute concurrently; the {\em for loop} of Line~3 thus contributes $2w(A)$ to circuit depth.
    Similarly, the gates of Lines~11-12 are non-overlapping and the {\em for loop} of Line~9 thus contributes $2w(B)$ to circuit depth.
    Adding one layer each of measurement/reset in Lines~7~and~13 (but excluding Lines~3-7 in iteration $l=0$) brings the circuit depth to $(2w(A)+2w(B)+2)d$.
    Finally, the ancilla preparation of Line~1 and gates of Lines~4-6 in iteration $l=0$ add $2w(A)+1$ of depth.
    For a visualization of the temporal sequence of operations, we refer the reader to \cref{fig:PackedCircuit} in the Appendix.
\end{proof}

\begin{remark}\label{remark:NonModular}
    While \cref{algorithm:sparse_cyclic_layout} is written with a particular hardware layout in mind, simply removing the cyclic shifts of Lines~4~and~10 yields a circuit suitable for an idealized monolithic device wherein {\em any} set of operations (including two-qubit gates) with non-overlapping support is allowed to occur concurrently.
    Moreover, the resulting circuit is {\em maximally packed}; meaning, it leaves no idle qubits at all between layers of gates.
    Accounting for depth in the same manner as in \cref{prop:SparseCyclicLayout}, the non-modular circuit has depth $(w(\polyA)+w(\polyB)+2)d + (w(\polyA)+1)$.
    Indeed, we used just such a syndrome extraction circuit for simulations under standard circuit noise in \cref{subsec:Performance}.
\end{remark}

\section{Conclusion}
In this work, we proposed a simple construction of hypergraph product codes based on classical cyclic codes.
While examples of cyclic HGP codes were previously considered in the literature ({\em e.g.} \cite{panteleev2021degenerate} introduces examples of \CTwo codes and \cite{siegel2024towards} discusses \CxR codes), this work to our knowledge is the first broad survey of cyclic HGP codes and led to the discovery of high performance instances.

Despite the simplicity of the construction, we were able to find cyclic HGP instances that exhibit code rates, minimum distances, and circuit level performance that are competitive with both BB codes and significantly better than previous state-of-the-art ML-optimized HGP codes.

We also proposed a cyclic layout alongside an even simpler cyclic HGP sub-family of codes, formed by the product with a repetition code.
This simpler sub-family boasts good circuit level performance, as well as being very hardware efficient under our cyclic layout -- an important attribute for real-world implementation.

Our results, combined with the flexible attributes of hypergraph product codes generally, suggest that it may yet remain a relevant candidate for long-term fault-tolerance.

This work focuses on the design of a code or a fault-tolerant quantum memory. We leave the design and optimization of logical gates for these codes for future work. One can immediately leverage previous constructions of logical gates for these codes~\cite{krishna2021fault, quintavalle2023partitioning, tiew2025low, xu2025fast, patra2025targeted, hong2025single, siegel2025snakes}.

\section{Acknowledgment}
The authors thank Pavel Panteleev, Ryan Tiew, Aharon Brodutch, Min Ye, Felix Tripier, John Gamble and the whole IonQ team for useful and insightful discussions.
 
\bibliography{references.bib}

\appendix
\section{Compressed Circuit}\label{appendix:PackedCircuit}

\Cref{fig:PackedCircuit} illustrates the syndrome extraction circuit described in \cref{algorithm:sparse_cyclic_layout}.
Each controlled-X and controlled-Z depicted in \cref{fig:PackedCircuit} are in fact {\em many} such gates, acting between the qubit sets $\q_{(*,*,0)}$, $\q_{(*,*,1)}$, and ancillae sets $\anc_{(*,*,X)}$, $\anc_{(*,*,Z)}$ (see \cref{sec:Layout} for indexing and notation).
Vertical dashed lines indicate time progression.

For HGP codes in general, the controlled-X and controlled-Z sub-circuits residing between times $t_2$ and $t_3$ arise from matrices $I\otimes B$ and $I\otimes B^T$ respectively, and contain exactly the same number of physical gates.
The same is true of the sub-circuits between $t_3$ and $t_4$, which arise from $A\otimes I$ and $A^T\otimes I$.

For \CxC codes specifically, the controlled-X and controlled-Z sub-circuits in each time step further has the property that they have identical depths, and are {\em maximally packed} (see \cref{remark:NonModular}) in the sense that they contain no idle qubits.

The blue dashed box in \cref{fig:PackedCircuit} indicates operations corresponding to one syndrome round.
Note that consecutive syndrome rounds are interleaved: controlled-X gates between $\q_{(0,*,*)}\leftrightarrow\anc_{(X,*,*)}$ for a subsequent syndrome round is performed in the same time step as controlled-Z gates between $\q_{(1,*,*)}\leftrightarrow\anc_{(Z,*,*)}$ for a previous syndrome round ({\em e.g.} between $t_3$ and $t_4$).

\begin{figure}
    \begin{centering}
    \includegraphics[width=8.5cm]{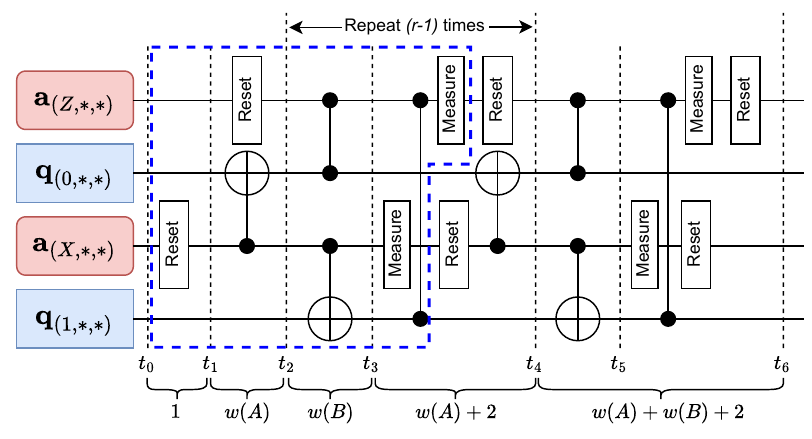}
    \par\end{centering}
    \caption{Packed syndrome extraction circuit for HGP.
    Each controlled-X and controlled-Z gate shown here is in fact many such physical gates.
    Here, $\q_{(*,*,*)}$ and $\anc_{(*,*,*)}$ are qubit and ancilla registers respectively (see \cref{sec:Layout} for indexing notation). Depth intervals shown do not include cyclic shifts.}
    \label{fig:PackedCircuit}
\end{figure}

\section{Fit Parameters}\label{appendix:FitParams}

\Cref{tab:FitParams} shows fit parameters for heuristic function for the logical error rate curves shown in \cref{fig:rep_circuitperf,fig:sym_circuitperf}.

\begin{table}[h]
    \begin{centering}
    \begin{tabular}{|c|c|c|c|c|}
    \hline 
    Code & $\alpha$ & $\beta (\times 10^{2})$ & $\gamma (\times 10^{3})$ & $\omega$ \tabularnewline
    \hline 
    \hline 
    {[}{[}450,32,8{]}{]} (\CTwo) & $16.22$ & $2.66$ & $4.336$ & 6\tabularnewline
    \hline 
    {[}{[}882,98,8{]}{]} (\CTwo) & $18.01$ & $6.05$ & $-72.88$ & 8\tabularnewline
    \hline 
    {[}{[}882,50,10{]}{]} (\CTwo) & $20.09$ & $2.71$ & $25.16$ & 6\tabularnewline
    \hline 
    {[}{[}336,20,6{]}{]} (\textsf{CxR}) & $12.02$ & $4.572$ & $-28.04$ & 6\tabularnewline
    \hline 
    {[}{[}336,14,8{]}{]} (\textsf{CxR}) & $15.10$ & $5.650$ & $-16.85$ & 6\tabularnewline
    \hline 
    {[}{[}240,8,8{]}{]} (\textsf{CxR}) & $14.30$ & $4.765$ & $-4.284$ & 5\tabularnewline
    \hline 
    {[}{[}420,10,10{]}{]} (\textsf{CxR}) & $17.89$ & $7.087$ & $-14.71$ & 5\tabularnewline
    \hline 
    {[}{[}620,20,10{]}{]} (\textsf{CxR}) & $16.89$ & $21.78$ & $-172.9$ & 7\tabularnewline
    \hline 
    {[}{[}625,25,8{]}{]} (ML-opt.) & $19.30$ & $8.299$ & $-134.4$ & $-$\tabularnewline
    \hline 
    \end{tabular}
    \par\end{centering}
    \caption{Fit parameters for logical error rate heuristic $\tilde{p}_{\log} = p^{d/2}e^{\alpha+\beta p+\gamma p^2}$.}\label{tab:FitParams}
\end{table}

\section{Code Table}\label{appendix:CodeTable}
\begin{table}[h]
\begin{centering}
\begin{tabular}{|c|c|c|c|c|c|}
\hline 
${\cal A}(x)$ & $n_{c}$ & $k_{c}$ & $d_{c}$\tabularnewline
\hline 
\hline 
$1+x+x^{4}$ & $15$ & 4 & 8\tabularnewline
\hline 
$1+x+x^{5}$ & $21$ & 5 & 10\tabularnewline
\hline 
$1+x^{2}+x^{4}+x^{10}$ & $28$ & 10 & 6\tabularnewline
\hline 
$1+x+x^{3}+x^{8}$ & $21$ & 7 & 8\tabularnewline
\hline 
$1+x+x^{2}+x^{7}$ & $30$ & 6 & 14\tabularnewline
\hline 
$1+x+x^{2}+x^{6}+x^{27}$ & $31$ & 10 & 10\tabularnewline
\hline 
$1+x+x^{3}+x^{9}+x^{10}$ & 31 & 10 & 12\tabularnewline
\hline 
\end{tabular}
\par\end{centering}
\caption{Table of polynomials defining classical cyclic codes, and their code parameters, on which a selection of quantum codes highlighted in the main text are based.}
\label{tab:CodeTableClassical}
\end{table}

\begin{table}[h]
\begin{centering}
\begin{tabular}{|c|c|c|c|c|c|}
\hline 
$n_{c}$ & $k_{c}$ & $d_{c}$ & \CTwo & \CxR\tabularnewline
\hline 
\hline 
$15$ & 4 & 8 & $[[450,32,8]]$ & $[[240,8,8]]$\tabularnewline
\hline 
$21$ & 5 & 10 & $[[882,50,10]]$ & $[[420,10,10]]$\tabularnewline
\hline 
$28$ & 10 & 6 & $[[1568,200,6]]${*} & $[[336,20,6]]$\tabularnewline
\hline 
$21$ & 7 & 8 & $[[882,98,8]]$ & $[[336,14,8]]$\tabularnewline
\hline 
$30$ & 6 & 14 & $[[1800,72,14]]$ & $[[840,12,14]]${*}\tabularnewline
\hline 
$31$ & 10 & 10 & $[[1922,200,10]]${*} & $[[620,20,10]]$\tabularnewline
\hline 
31 & 10 & 12 & $[[1922,200,12]]$ & $[[744,20,12]]${*}\tabularnewline
\hline 
\end{tabular}
\par\end{centering}
\caption{Table of corresponding \CTwo and \CxR quantum codes for each classical cyclic code.
Here, asterisks (*) denote a quantum code instance that was not explicitly discussed in main text either due to excessive block size or poorer distance / encoding rate than comparable codes in the table.}
\label{tab:CodeTableQuantum}
\end{table}

\Cref{tab:CodeTableClassical} show generating polynomials for a selection of classical cyclic codes, on top of which quantum \CTwo and \CxR codes highlighted in the text were built.
Therein, $\cal{A}$ and $x$ have the meaning of \cref{def:cHGP}.
In \cref{tab:CodeTableQuantum}, for each classical cyclic code we show parameters of the corresponding \CTwo and \CxR quantum codes.

\end{document}